\documentclass[11pt]{article}

\usepackage{amssymb}
\usepackage{amsmath, amssymb, amsthm, amsfonts, dsfont, bbm}
\usepackage{enumerate}
\usepackage{url}

\def\phi{\varphi}

\newcommand{\mathrmL}{{\mathchoice{\mbox{\rm\L}}{\mbox{\rm\L}}{\mbox{\rm\scriptsize\L}}{\mbox{\rm\tiny\L}}}}

\renewcommand{\lnot}{\mathord{\sim}}

\newtheorem{theorem}{Theorem}[section]
\newtheorem{definition}[theorem]{Definition}
\newtheorem{lemma}[theorem]{Lemma}

\newtheorem{remark}[theorem]{Remark}
\newtheorem{corollary}[theorem]{Corollary}

\newtheorem{fact}[theorem]{Fact}

\newtheorem{example}[theorem]{Example}

\newcommand{\lang}[1]{\ensuremath{\mathcal #1}}
\newcommand{\Fm}{\ensuremath{\mathit{Tm}}}
\newcommand{\Tm}{\Fm}

\newcommand{\Vbl}{\ensuremath{\mathrm{Var}}}

\newcommand{\logic}[1]{\ensuremath{\mathrm {#1}}}
\newcommand\FL[1]{\ensuremath{ \logic{{FL}}_\mathrm{#1} }}
\newcommand\FLew{\FL{ew}}
\newcommand\WCon{\logic{WCon}}
\newcommand\CL{\ensuremath{\mathrm{CL}}}
\newcommand\MTL{\logic{MTL}}
\newcommand\SMTL{\logic{SMTL}}
\newcommand\BL{\logic{BL}}
\newcommand\SBL{\logic{SBL}}
\newcommand\Int{\logic{Int}}

\newcommand{\alg}[1]{ {\ensuremath{\mathcal {#1}}}}
\newcommand{\Alg}[1]{\ensuremath{\mathbbmss {#1}}}
\newcommand{\aFLew}{\Alg{FL_{ew}}}
\newcommand{\BA}{\Alg{BA}}
\newcommand{\HA}{\Alg{HA}}
\newcommand{\aWCon}{\Alg{WCon}}
\newcommand{\MV}{\Alg{MV}}
\newcommand{\aBL}{\Alg{BL}}
\newcommand{\aMTL}{\Alg{MTL}}
\newcommand{\aG}{\Alg{G}}

\newcommand{\aInFLew}{\Alg{InFL_{ew}}}

\newcommand\standardL{\ensuremath{[0,1]_\mathrmL}}
\newcommand\standardP{\ensuremath{[0,1]_{\Pi}}}
\newcommand\standardG{\ensuremath{[0,1]_{\mathrm G}}}

\newcommand\stBool{ \ensuremath{ {\{0,1\}_\mathrm{B} }} }

\newcommand\TAUT{\ensuremath{\mathrm{TAUT}}}

\newcommand\SAT{\ensuremath{\mathrm{SAT}}}
\newcommand\SATPOS{\ensuremath{\mathrm{SATPOS}}}
\newcommand\oSAT{\ensuremath{\mathrm{\overline{SAT}}}}
\newcommand\oSATPOS{\ensuremath{\mathrm{\overline{SATPOS}}}}

\newcommand\ComplexityClass[1]{\ensuremath{\mathrm{#1}}}
\newcommand\Po{\ensuremath{\ComplexityClass{P}}}
\newcommand\NP{\ensuremath{\ComplexityClass{NP}}}
\newcommand\coNP{\ensuremath{\ComplexityClass{coNP}}}

\newcommand\DP{\ensuremath{\ComplexityClass{DP}}}

\begin{document}

\title{Term satisfiability in $\FLew$-algebras\footnote{http://dx.doi.org/10.1016/j.tcs.2016.03.009}
\footnote{\copyright\ 2016. This manuscript version is made available under the CC-BY-NC-ND 4.0 license http://creativecommons.org/licenses/by-nc-nd/4.0/.}}
\author{Zuzana Hanikov\'a and Petr Savick\'y\\
Institute of Computer Science, Czech Academy of Sciences,\\ 182 07 Prague, Czech Republic\\
hanikova@cs.cas.cz, savicky@cs.cas.cz}
\date{}
\maketitle

\begin{abstract}
$\FLew$-algebras form the algebraic semantics of the full Lambek calculus with exchange and weakening.
We investigate two relations, called \emph{satisfiability} and \emph{positive satisfiability}, 
between $\FLew$-terms and $\FLew$-algebras.  
For each $\FLew$-algebra, the sets of its satisfiable and positively satisfiable terms 
can be viewed as fragments of its existential theory;
we identify and investigate the complements as fragments of its universal theory. 
We offer characterizations of those algebras that (positively) satisfy just those terms
that are satisfiable in the two-element Boolean algebra providing its semantics to classical propositional logic.
In case of positive satisfiability, these algebras are just the nontrivial weakly contractive $\FLew$-algebras.
In case of satisfiability, we give a characterization by means of another property of the algebra, 
the existence of a two-element congruence.
Further, we argue that (positive) satisfiability problems in $\FLew$-algebras are computationally hard.
Some previous results in the area of term satisfiability in MV-algebras or BL-algebras 
 are thus brought to a common footing with known facts on satisfiability in Heyting algebras.
\end{abstract}

\section{Introduction}

This work investigates two satisfiability relations between 
terms of a particular algebraic language and 
algebras interpreting that language. 
It often refers to Boolean term satisfiability, 
its semantic setting is however broader,
and one of its main aims is to determine whether or not,
and how, this broader setting in fact extends 
the set of terms satisfiable in the two-element Boolean algebra.\footnote{In this paper, 
the interpretation provided by the two-element Boolean algebra is referred to shortly as `classical'.}
As our base, we choose  a class of algebras that forms the equivalent algebraic semantics 
of the full Lambek calculus with exchange and weakening (traditionally denoted $\FLew$). 
This propositional logic is regarded as a \emph{substructural} logic
in the sense of \cite{Schroeder-Heister-Dosen:Substructural, Galatos-JKO:ResiduatedLattices}. 
The framework of substructural logics 
is a means of bringing many different logical systems to a common denominator. 
In particular, all of the logics are considered in the same language.\footnote{As a consequence,
some function symbols in this language become term-definable in some of the stronger logics, such as classical logic.}  
Extensions of the logic $\FLew$ include classical propositional logic, 
as well as the intuitionistic and superintuitionistic logics, 
the monoidal t-norm logic $\MTL$, H\'ajek's logic $\BL$, and {\L}ukasiewicz logic $\L$;  
likewise, the class of $\FLew$-algebras is quite comprehensive.   

We work with the algebraic semantics of $\FLew$. 
We do not use its formal deductive systems, 
nor do we present any here.
No distinction is made between algebraic terms and propositional formulas and
the logic is introduced algebraically, that is, it is formally identified with the set of 
terms that are valid under all interpretations given by $\FLew$-algebras.
An $\FLew$-algebra\footnote{Cf.~Definition \ref{def:FLew-algebra}.} $\alg{A}$ 
has the partial order of a bounded lattice, with least element $0^\alg{A}$,  greatest element $1^\alg{A}$, 
and two lattice operations $\wedge^\alg{A}$ and $\vee^\alg{A}$.
Further there is a commutative residuated monoidal operation $\cdot^\alg{A}$, with neutral element $1^\alg{A}$
and residuum $\to^\alg{A}$. 
A term is \emph{tautologous} (valid) in an $\FLew$-algebra $\alg{A}$
iff all $\alg{A}$-assignments send it to $1^\alg{A}$,
whereas it is \emph{satisfiable} in $\alg{A}$ iff 
the same occurs under some $\alg{A}$-assignment,
and \emph{positively satisfiable} iff some $\alg{A}$-assignment sends it to an element greater than $0^\alg{A}$. 

In other works, the term `satisfiability' may relate not just to terms but rather
to first-order formulas, and what is in fact investigated is the existential theory of the structures (algebras). 
Needless to say, that approach is more general, subsuming ours as one of its fragments. 
Still, in Section \ref{section:MV} we reformulate a theorem of Gispert, 
as given in \cite{Gispert-Mundici:MVAlgebras}, to show that in an MV-chain (a linearly
ordered MV-algebra),  
tautologousness and satisfiability fully determine its universal (and hence existential) theory. 

Why $\FLew$-algebras? As we intend to make comparisons to classical 
satisfiability, we wish to preserve its flavour (despite the new semantics of terms);
it is perhaps best preserved with  algebras that interpret the (full) classical language, 
carry an order, and have a least element and a greatest element in this order.
This demand cautions not to settle for a too broad class of algebras. 
At the same time, there is a call for a comprehensive class: there are previous results 
not only on Boolean satisfiability and Heyting satisfiability, but also on satisfiability 
in MV-algebras, G\"odel-Dummett algebras and product algebras (\cite{Mundici:Satisfiability, Hajek:1998}),
and we want to be able to relate to these results. 
Apart from these two demands that seem to balance each other,
one wants to be somewhat familiar with the class one works with; 
rather a lot is known about the lattice of subvarieties 
of the variety of $\FLew$-algebras and about some of the subvarieties in themselves (\cite{Galatos-JKO:ResiduatedLattices}).

This paper can be read as a study of the features satisfiability acquires 
when one departs from the classical interpretation;
despite the many differences, we show (positive) satisfiability in $\FLew$-algebras 
ties to classical satisfiability in several ways.
Let us briefly reflect on the analogies and the distinctions we are facing.
As remarked, the class of $\FLew$-algebras subsumes 
the class of Boolean algebras.
Importantly, the two-element Boolean algebra is a subalgebra of 
each nontrivial $\FLew$-algebra, obtained by considering 
just the least and the greatest element of the bounded lattice
with the restricted operations. Therefore, a term 
that is classically satisfiable is also satisfiable in any nontrivial $\FLew$-algebra.
Under a fixed interpretation provided by the two-element Boolean algebra, 
tautologousness and satisfiability are just 
\emph{properties} of terms, and they are related:
a term is a classical tautology iff its negation\footnote{where $\neg \varphi$, 
the negation of a term $\varphi$, is defined as $\varphi\to 0$} is not classically satisfiable,
and vice versa. 
Both tautologousness and satisfiability depend essentially on the algebraic interpretation: 
within this paper, the algebra interpreting the language is not fixed,
but is viewed as an argument to a satisfiability operator, while 
satisfiability itself is considered as a binary relation between algebras and terms.
The interpretation does make a difference:
a prime example is the well-known standard MV-algebra on the real unit interval [0,1],
interpreting the infinite-valued logic of {\L}ukasiewicz, which satisfies terms that
are classically unsatisfiable, such as 
 $x \equiv \neg x$ under the assignment $x\mapsto 1/2$. 
Satisfiability and positive satisfiability of terms 
in the standard MV-algebra has been investigated by Mundici in \cite{Mundici:Satisfiability},
who has shown  both problems to be $\NP$-complete.
The result is exciting: the domain of the standard MV-algebra is infinite
(it has the cardinality of the continuum),
thus there is no obvious way 
either of testing satisfiability of a term or of certifying it succinctly. 
Indeed for many subalgebras of the standard MV-algebra
satisfiability of terms is algorithmically undecidable;
we show this in Section~\ref{section:MV}. 

As another well-known example, one may consider the class of Heyting algebras,
the algebraic counterpart of intuitionistic logic. 
It is well known (cf.~\cite{Chagrov-Zakharyaschev:ModalLogic}) 
that subvarieties of Heyting algebras present a rich structure,
so in this class, tautologousness is clearly a \emph{relation} 
between terms of the above language and algebras of the given class:
considering different Heyting algebras, the set of their tautologies may differ.
Yet it is well known by Glivenko theorem (\cite{Glivenko:Equivalence}) 
that a term is satisfiable in a nontrivial Heyting algebra
iff it is classically satisfiable.
These two facts make it apparent that the link, familiar from classical logic, 
between satisfiability and tautologousness is missing for Heyting algebras; 
despite the fact that there is a vast number of logics/sets of tautologies
given by Heyting algebras, the sets of satisfiable terms for any two nontrivial Heyting algebras coincide.

The observation that satisfiability need not link simply to tautologousness, under a semantics
more general than the classical one, 
prompts a study of satisfiability in its own right. 
The classical link, occasioned by the duality of quantifiers, is preserved when,
for an $\FLew$-algebra under consideration, or indeed for any first-order structure, 
one looks not at its terms but at its full existential theory and its full universal theory: 
an existential sentence $\Phi$ is in the existential theory of the structure 
iff $\lnot\Phi$ (a universal sentence, where $\lnot$ denotes the classical negation) 
is not in the universal theory of that structure.
Still, tautologousness is a tiny fragment of the universal theory and satisfiability 
is a tiny fragment of the existential theory, 
so the duality need not be reasonably helpful regarding 
what can be said of the relation of these two fragments.
To reconstruct part of the duality, we identify the fragments of universal
theory that complement satisfiability and positive satisfiability; this is done in Section \ref{section:satisfiability}.

Under a given interpretation, both tautologousness and satisfiability of terms 
constitute a decision problem, and one may ask how difficult it is to recognize the set of such terms. 
Tautologousness/theoremhood problems for logics extending $\FLew$, including their decidability 
and computational complexity, have merited a lot of attention, while
satisfiability studies (apart from the classical SAT problem,
which is the standard NP-complete problem, cf.~\cite{Cook:Complexity1971}) are more scarce.
Works in term satisfiability, particularly in its computational complexity,
for $\FLew$-algebras include Mundici's work \cite{Mundici:Satisfiability}  
showing $\NP$-completeness of satisfiability and positive satisfiability in the standard MV-algebra and
results in H\'ajek's book \cite{Hajek:1998}, showing that satisfiability in 
the standard G\"odel and product algebras is classical and hence $\NP$-complete; 
the $\NP$-completeness results are extended to any standard BL-algebra in \cite{Hanikova:thesis}.
\cite{Cintula-Hajek:ComplexityLukasiewicz} shows $\NP$-completeness for satisfiability
in axiomatic extensions of {\L}ukasiewicz logic. 
Recently, the  paper \cite{vanAlten:PartialAlgebras} has addressed the application of partial algebras to the existential theory
of Boolean and Heyting algebras.\footnote{The last two mentioned results
are concerned with satisfiability in \emph{classes} of algebras.}

One of the two following definitions of the classical satisfiability problem is usually considered:
\begin{align}
 \SAT & = \{\varphi \mid  \alg{A} \models \exists \bar x (\varphi(\bar x)\approx 1 )\} \label{FL_SAT1_BOOL} \\
 \SAT & = \{\varphi \mid \alg{A} \models \exists \bar x (\varphi(\bar x) >0 ) \} \label{FL_SAT0_BOOL}
\end{align}
where  $\varphi$ ranges over well-formed terms of the language. 
These two definitions yield, interpreted classically, the same set of terms.
Given an $\FLew$-algebra $\alg{A}$, one can distinguish 
\begin{itemize}
\item (fully) satisfiable terms: as in (\ref{FL_SAT1_BOOL});
\item positively satisfiable terms: as in (\ref{FL_SAT0_BOOL});
\item unsatisfiable terms:  the complement of  (\ref{FL_SAT0_BOOL}).
\end{itemize}
In a nontrivial $\FLew$-algebra, 
positively satisfiable terms subsume (fully) satisfiable ones.
If $\alg{A}$ is distinct from the two-element Boolean algebra, 
it is often not obvious whether or not 
(\ref{FL_SAT1_BOOL}) and (\ref{FL_SAT0_BOOL}) yield the same set of terms, i.e., 
whether there are any terms that are neither fully satisfiable nor unsatisfiable in $\alg{A}$.
This is one of the points addressed in this paper.

\medskip
{\bf Paper structure overview.} 
Section \ref{section:prelim} defines $\FLew$-algebras and introduces some subvarieties
that are of interest within this paper.
Section \ref{section:satisfiability} introduces the binary relations of 
full satisfiability and positive satisfiability between
$\FLew$-terms and $\FLew$-algebras and discusses some of their properties.  
Section \ref{section:satpos} focuses on positive satisfiability. It gives a characterization 
of those $\FLew$-algebras where positive satisfiability is classical (Theorem \ref{theorem:char_satpos}).  
It turns out that those are exactly the nontrivial \emph{weakly contractive algebras} within $\FLew$; 
the variety of weakly contractive algebras can be delimited by a single identity within $\FLew$. 
Moreover, the class of nontrivial weakly contractive $\FLew$-algebras characterizes those algebras 
where full satisfiability coincides with positive satisfiability.
We also show that there are continuum many different positive satisfiability problems
for $\FLew$-algebras.
Section \ref{section:fulsat} is about full satisfiability; 
it gives a characterization of $\FLew$-algebras with classical full satisfiability
in terms of another property, the existence of a congruence with exactly two classes (Corollary \ref{corSATvsHom}).
Moreover, for $\FLew$-chains, nonclassical full satisfiability can be 
characterized using a single term (Theorem \ref{thChainsClassicalSAT}).
Section \ref{section:MV} discusses satisfiability in MV-algebras, 
showing, i.a., that the set of the full satisfiability problems
given by MV-chains (a fortiori, by $\FLew$-algebras) has the cardinality of the continuum (Theorem \ref{theorem:MV_SATs}).
Some fragments of the $\FLew$-language are considered in Section \ref{section:CNFandDNF}, 
and it is shown that in a nontrivial $\FLew$-algebra, satisfiability and positive satisfiability are $\NP$-hard (Corollary \ref{CF-NP-complete}).
In Section \ref{section:DP} we look at the difference of the set of positively satisfiable terms
and the set of fully satisfiable ones in an $\FLew$-algebra; 
we show that, if nonempty, this set is $\DP$-hard (Theorem \ref{th_DP_hard}); 
if, moreover, both the full satisfiability and the positive satisfiability problems
 are in $\NP$, then it is  $\DP$-complete.

\section{Preliminaries}
\label{section:prelim}

This section introduces the class of $\FLew$-algebras along with some subclasses.  

This paper is concerned with the algebraic semantics of propositional logics, 
therefore, speaking of a \emph{logic} 
(such as classical, intuitionistic, or $\FLew$), 
what is meant is just the propositional part thereof.
Given an algebraic language $\lang{L}$ (a set of function symbols) containing $\to$, 
a \emph{logic} in the language $\lang{L}$ is a set of $\lang{L}$-terms 
that is substitution invariant and closed under logical consequence 
(the modus ponens rule). 
Given two logics $\logic{L,L'}$ in a language $\lang{L}$, 
one says $\logic{L'}$ \emph{extends} $\logic{L}$
iff $\logic{L}\subseteq \logic{L'}$. 
A logic $\logic{L}$ in a language $\lang{L}$ is \emph{consistent} 
iff it is distinct from the set of all $\lang{L}$-terms,
otherwise it is inconsistent. 

As remarked, this paper makes no distinction between 
logical connectives and function symbols 
nor between formulas and terms. 
Because our setting is algebraic, our preference is the latter respectively.
The language of $\FLew$ has four binary function symbols: 
$\cdot$ (called \emph{multiplication}, or \emph{multiplicative conjunction}), 
$\to$ (\emph{implication} or \emph{residuation}), 
$\wedge$ and $\vee$ (\emph{lattice conjunction} and \emph{disjunction}),  
and two constants $0$ and $1$. 

Moreover, a countably infinite set of variables is considered:  
$\Vbl=\{x_i\}_{i\in N}$, where elements are informally denoted with lowercase letters such as $x,y,z$.
An $n$-tuple $x_1,\dots, x_n$ of variables may be denoted $\bar x$.   
$\FLew$-terms are defined inductively as usual, 
and denoted with lowercase Greek letter such as $\varphi, \psi, \chi$. 
The notation $\varphi(x_1,\dots, x_n)$ (or $\varphi(\bar x)$) signifies that
all the variables occurring in the term $\varphi$ are among $x_1,\dots, x_n$ (or $\bar x$). 
The set of all terms of the language of $\FLew$ is denoted  $\Fm$.
All terms within this paper are implicitly considered $\FLew$-terms,
unless stated otherwise.

The unary symbol $\neg$ (negation) is introduced 
by writing $\neg \phi$ for $\phi\to 0$ for any term $\phi$; 
moreover, we write $\phi \equiv \psi$ for $(\phi\to\psi)\cdot(\psi\to\phi)$
and $\varphi+\psi$ for $\neg(\neg \varphi \cdot \neg \psi)$
for any pair of terms $\phi$ and $\psi$.
Precedence of function symbols is as follows:
$\neg$ binds stronger than any binary symbol; 
$\cdot$, $\wedge$ and $\vee$ bind stronger than $\to$ and $\equiv$.
For a term $\varphi$, we write $\varphi^n$ for $\varphi\cdot\varphi\cdot\dots\cdot\varphi$ ($n$ terms)
and $n\varphi$ for $\varphi+\varphi+\dots+\varphi$ ($n$ terms).

An interpretation of a function/predicate symbol $f$ in an algebra $\alg{A}$ is denoted $f^\alg{A}$.
We use $\approx$ as the identity symbol and $=$ 
for equality in an algebra $\alg{A}$. 
The superscripts may be omitted if no confusion can arise.

Moreover, as we work with fragments of algebraic theories,
we need notation for connectives of classical logic, occurring in first-order algebraic formulas
(whose atoms are algebraic identities):
we shall use $\&$ for the conjunction, $\Rightarrow$ for the implication,
$\lnot$ for the negation, and $\bot$ for falsity.

\begin{definition}
\label{def:FLew-algebra}
An algebra $\alg{A} = \langle A, \cdot^\alg{A}, \to^\alg{A}, \wedge^\alg{A}, \vee^\alg{A}, 0^\alg{A}, 1^\alg{A} \rangle$
with four binary operations and two constants is an $\FLew$-algebra if
\begin{itemize}
\item[(1)] $\langle A, \wedge^\alg{A}, \vee^\alg{A}, 0^\alg{A}, 1^\alg{A} \rangle$ is a bounded lattice with the least element $0^\alg{A}$ and the greatest element $1^\alg{A}$; we use $\leq^\alg{A}$ for the lattice order;
\item[(2)] $\langle A, \cdot^\alg{A}, 1^\alg{A} \rangle $ is a commutative monoid with the unit element $1^\alg{A}$; 
\item[(3)] $\cdot^\alg{A}$ and $\to^\alg{A}$ form a residuated pair, i.e., 
$x\cdot^\alg{A} y\leq^\alg{A} z$ iff $x\leq^\alg{A} y\to^\alg{A} z$.
\end{itemize}
\end{definition} 

\begin{remark}{\rm 
The acronym $\FLew$, standing for Full Lambek calculus with exchange and weakening, 
indicates that $\FLew$ can be obtained as an extension of another logic---namely, 
the full Lambek calculus, $\FL{}$---with two axioms/rules: \emph{exchange} (yielding commutativity  of $\cdot$)
and \emph{weakening} (vouchsafing that the lattice order is bounded with $0$ as bottom and\/ $1$ as top).
Since exchange and weakening can be rendered as two of three structural rules in a particular 
sequent calculus for intuitionistic logic (the remaining rule being contraction),
the logics obtained from this calculus by removing some of the structural rules, and their axiomatic extensions,
are called \emph{substructural}. 
Namely, the logic $\FLew$ is often mentioned as a ``contraction-free'' logic,
 as opposed to intuitionistic logic, which can be rendered as $\FL{ewc}$.
}
\end{remark}

\begin{example}{\rm 
In a commutative unital ring, 
the set of its ideals, endowed with the inclusion order (which yields a modular lattice), 
ideal multiplication, and the corresponding residuum, is an $\FLew$-algebra; cf.~\cite{Galatos-JKO:ResiduatedLattices}.
The $0$-free reduct of this algebra was an important example of a \emph{residuated lattice} 
as originally considered in \cite{Ward-Dilworth:ResLat}.\footnote{In \cite{Galatos-JKO:ResiduatedLattices},
a residuated lattice is the $0$-free reduct of an $\FL{}$-algebra.} 
}
\end{example}

The paper \cite{Ono:withoutContraction} is dedicated to $\FLew$ and its extensions.
The logic was also studied in \cite{Hohle:Monoids}.
See \cite{Galatos-JKO:ResiduatedLattices} for a development of $\FLew$ inside the substructural logic landscape.
$\FLew$-algebras can be shown to form a variety of algebras, which will be denoted  $\aFLew$.

Residuation entails  
$\langle A, \cdot^\alg{A}, 1^\alg{A}, \leq^\alg{A} \rangle$
is a partially ordered monoid, i.e., $\cdot^\alg{A}$ preserves the order.
An $\FLew$-algebra $\alg{A}$ is \emph{linearly ordered}, or, a \emph{chain},
whenever $\leq^\alg{A}$ is a linear order on $A$. 

\begin{definition}
Let $\alg{A}$ be an $\FLew$-algebra and $\varphi$ an $\FLew$-term.
The term $\varphi$ is a \emph{tautology} of $\alg{A}$ 
iff $\alg{A} \models \forall \bar x ( \varphi(\bar x)\approx 1)$.
The set of all tautologies of $\alg{A}$ is denoted\/ $\TAUT(\alg{A})$.
\end{definition}

For a class $\Alg{K}$ of $\FLew$-algebras, 
the term $\varphi$ is a tautology of $\Alg{K}$  
iff it is a tautology of each $\alg{A}\in \Alg{K}$. 
In case an $\FLew$-term $\varphi$ is a tautology of an $\FLew$-algebra $\alg{A}$,
we also say that $\varphi$ is \emph{valid} in $\alg{A}$ 
or that it \emph{holds} in $\alg{A}$,
and write simply $\alg{A}\models\varphi$.
For each $\alg{A}$, the set $\TAUT(\alg{A})$ is a logic,
often called \emph{the logic of $\alg{A}$} and denoted $\logic{L}(\alg{A})$.

\begin{definition} 
The logic $\FLew$ is the set of\/ $\FLew$-terms that are tautologies in each $\FLew$-algebra.
\end{definition}

We forfeit introducing a deductive system for the logic $\FLew$; 
see \cite{Galatos-JKO:ResiduatedLattices, Ono:withoutContraction} and references therein.
Thus we implicitly rely on completeness theorems for $\FLew$, which follow from algebraizability.

\medskip

We list some statements on $\FLew$-algebras (the superscripts denoting interpretation
are omitted for the sake of readability).

\begin{fact}
\label{FLew_facts}
Let $\alg{A}$ be an $\FLew$-algebra and $x,y,z\in A$. 
\begin{itemize} 
\item[(1)] $x\leq y$ iff $x\to y = 1$; in particular (taking $0$ for $y$), $\neg x = 1$ iff $x = 0$. 
\item[(2)] $x \cdot (x\to y) \leq  y$; in particular (taking $0$ for $y$), $x\cdot\neg x \leq 0$ and $x \leq \neg\neg x$.
\item[(3)] $x\to y \leq (y \to z) \to (x \to z)$ (transitivity of $\to$); in particular (taking $0$ for $z$), 
    $x \to y \leq \neg y \to \neg x$. This yields $\neg x = \neg\neg\neg x$.
\item[(4)] $x \leq y \to x$ (weakening).
\end{itemize}
\end{fact}

\noindent
 On the other hand, for a given $x\in A$, 
$\neg x=0$ need not imply $x=1$.

\begin{fact}
\label{dist_law}
In an $\FLew$-algebra $\alg{A}$,
multiplication distributes over existing joins, i.e.,
if $\bigvee_{i\in I} y_i$ exists in $A$ for a nonempty $I$, where $y_i \in A$ for each $i\in I$, 
then for each $x\in A$, $\bigvee_{i\in I} x\cdot y_i$ exists also and
$ x\cdot \bigvee_{i \in I} y_i  = \bigvee_{i \in I}(x\cdot y_i)$.
\end{fact}
\medbreak

\begin{definition}
For a logic $\logic{L}$ extending the logic $\FLew$, an $\logic{L}$-algebra is an $\FLew$-algebra such that
all terms $\varphi\in \logic{L}$ are tautologies of $\alg{A}$.
\end{definition}

Logics extending $\FLew$ form a complete lattice ordered by inclusion, where the bottom is the logic
$\FLew$ and the top is the inconsistent logic $\Fm$.
This complete lattice structure also exists on the (dually isomorphic) lattice of subvarieties of $\aFLew$,
where the bottom is the class of trivial (one-element) $\FLew$-algebras and the top is the whole variety.

The following translations between terms and identities provide algebraizability of $\FLew$:
\begin{itemize}
\item $\varphi\mapsto\varphi\approx 1$, 
\item $\varphi\approx \psi \mapsto \varphi\equiv \psi$, 
\end{itemize}
where $\varphi$ and $\psi$ are $\FLew$-terms.

We have mentioned that a logic $\logic{L}$ is consistent iff it differs from the set of all terms in the language.
Since $0\to\varphi$ is an $\FLew$-tautology for any $\varphi$, 
one can define consistency of logics extending $\FLew$ by the condition that they do not contain the term $0$.

\medskip

Some important classes of $\FLew$-algebras are introduced below, with the corresponding logics;
 references are given.

\begin{itemize}

\item \emph{$\MTL$-algebras}, also known as semilinear $\FLew$-algebras, 
form a subvariety of $\aFLew$ delimited by the identity $(x\to y) \lor (y\to x) \approx 1$.
This variety is generated by $\FLew$-chains.
The logic $\MTL$, with the variety of $\MTL$-algebras that 
forms its equivalent algebraic semantics, was introduced 
by Esteva and Godo in \cite{Esteva-Godo:Monoidal}.

\item \emph{Weakly contractive $\FLew$ algebras}, $\aWCon$, form a subvariety of $\aFLew$ 
delimited by any of the following identities (equivalent over $\FLew$):
\begin{align*}
\neg (x^2) &\approx \neg x \\
x\to \neg x &\approx \neg x \\
x\land \neg x &\approx 0
\end{align*}
This subvariety is considered by Ono in  \cite{Ono:withoutContraction}.
The term $\SMTL$-algebras is used for weakly contractive $\MTL$-algebras.
 
\item \emph{Heyting algebras}, $\HA$, form a subvariety of $\aWCon$
delimited by the identity $x\cdot x \approx x$ 
(this identity also delimits $\HA$ within $\aFLew$).
In a Heyting algebra, the operations $\wedge$ and $\cdot$ coincide;
it follows that the bounded lattice order in a Heyting algebra is always a distributive lattice order,
and complete distributive lattices satisfying the distributive law given in Fact \ref{dist_law}
(taking $\wedge$ for $\cdot$) provide examples of Heyting algebras (in particular, 
any finite distributive lattice can be expanded to a Heyting algebra).
Heyting algebras form the equivalent algebraic semantics of intuitionistic logic $\Int$
(cf.~\cite{Chagrov-Zakharyaschev:ModalLogic} and references therein).

\item \emph{Boolean algebras}, $\BA$, 
form a subvariety of $\HA$ delimited by the identity $x\vee \neg x\approx 1$
(this identity also delimits $\BA$ within $\aFLew$).
Another identity that delimits $\BA$ within $\HA$
is the involutive law, $\neg\neg x \approx x$.  
Up to an isomorphism, there is just one Boolean algebra with exactly two distinct elements:
this is the algebra giving semantics to classical propositional logic $\CL$,
and throughout this paper it is referred to as \emph{the two-element Boolean algebra} and
denoted $\stBool$.
It generates the variety $\BA$. 

\item \emph{Involutive $\FLew$-algebras}, $\aInFLew$, form a subvariety of $\aFLew$
delimited by the identity $\neg\neg x \approx x$.
In an involutive $\FLew$-algebra $\alg{A}$, $\neg^\alg{A}$ is as an order-reversing bijection on $A$.
As $x+y$ stands for $\neg(\neg x \cdot \neg y)$,
in each involutive $\FLew$-algebra
$x + y$ is equivalent to  $\neg x\to y$, and $x\to y$ is equivalent to $\neg x +y$. 
Moreover, $x\cdot y$ is equivalent to $\neg (\neg x + \neg y)$.

\item \emph{$\BL$-algebras}, $\aBL$, form a subvariety of $\aMTL$ delimited by the identity
$x\land y \approx x\cdot(x\to y)$. 
They form the equivalent algebraic semantics of H\'ajek's logic $\BL$ (cf.~\cite{Hajek:1998}).
\emph{$\SBL$-algebras} are weakly contractive $\BL$-algebras.

\item \emph{MV-algebras}, $\MV$, are involutive $\BL$-algebras.
This variety forms the equivalent algebraic semantics of {\L}ukasiewicz logic 
(cf.~\cite{Cignoli-Ottaviano-Mundici:AlgebraicFoundations,DiNola-Leustean:Handbook} for references). 
The variety $\MV$ is generated by its single element, the 
\emph{standard MV-algebra} $\standardL$: the domain is $[0,1]$ and
the lattice order is the usual order of reals, while
the operations interpreting the language of $\FLew$ are, for each $x,y\in [0,1]$, as follows: 
$x\cdot y = \max(0,x+y-1)$; $x\to y = \min(1, 1-x+y)$;
$\neg x = 1-x$.
Other example of MV-algebras include the \emph{Komori chains} $\alg{K}_{n+1}$, for $n\geq 1$:
the domain of $\alg{K}_{n+1}$ is the interval $[ \langle 0,0\rangle, \langle n,0\rangle ]$ in the group 
$\mathds{Z} \times_{\rm lex} \mathds{Z}$; cf.~\cite{Komori:SuperLukasiewiczPropositional, DiNola-Leustean:Handbook}.  
In particular, $\alg{K}_2$ is the \emph{Chang algebra}. 

MV-algebras also satisfy the identity $x\cdot (\neg x+y) \approx y\cdot (\neg y +x)$.

\item \emph{G\"odel-Dummett algebras}, $\aG$, are semilinear Heyting algebras.
The variety is  generated by linearly ordered elements of $\HA$ and, in fact,
by a single element $\standardG$, the standard G\"odel-Dummett algebra.

\item \emph{product algebras} form a subvariety of $\SBL$-algebras, 
given by the identity 
$(x\to z)\lor(( x\to(x\cdot y))\to y) \approx 1$.
This variety is generated by the  standard product algebra,
with domain $[0,1]$ and $\cdot$ interpreted as multiplication.    
See \cite{Hajek:1998} and references therein for G\"odel-Dummett and product algebras.

\end{itemize}

\begin{fact}
\label{fact_inv_wcon}
{\rm (\cite{Ono:withoutContraction})}
 $\aWCon\cap\aInFLew=\BA$.
\end{fact}

\section{On satisfiability and positive satisfiability}
\label{section:satisfiability}

We define two binary relations $\SAT$ and $\SATPOS$ on $\aFLew \times \Fm$
and we discuss their basic properties.

\begin{definition} 
Let $\alg{A}$ be an $\FLew$-algebra and $\varphi$ an $\FLew$-term. Then\footnote{Recall that $\lnot$ denotes the Boolean negation in the algebraic theory.} 
\begin{align*}
\SAT(\alg{A}, \varphi) &\mbox{ iff } \alg{A} \models \exists \bar x (\varphi(\bar x) \approx 1)\\
\SATPOS(\alg{A}, \varphi) &\mbox{ iff } \alg{A} \models \exists \bar x \lnot (\varphi(\bar x) \approx 0)
\end{align*}
\end{definition}

For an $\FLew$-algebra $\alg{A}$, write 
\begin{align*}
\SAT(\alg{A}) &= \{ \varphi \mid \SAT ( \alg{A}, \varphi) \} \\
\SATPOS(\alg{A}) &= \{ \varphi \mid \SATPOS ( \alg{A}, \varphi) \}   
\end{align*}
In an analogous way, one might define $\SAT(\varphi)$ and $\SATPOS(\varphi)$ for an $\FLew$-term $\varphi$; 
the statements $\SAT(\alg{A}, \varphi)$; $\varphi \in \SAT(\alg{A})$; $\alg{A}\in \SAT(\varphi)$ are equivalent,
and analogously for $\SATPOS$.
$\SAT$ and $\SATPOS$ are used throughout this paper 
as unary operators, producing sets of $\FLew$-terms.

For an $\FLew$-algebra $\alg{A}$, the sets $\SAT(\alg{A})$ and $\SATPOS(\alg{A})$ are referred to as
(fully) \emph{satisfiable} and  \emph{positively satisfiable} terms of $\alg{A}$, respectively.
The set-theoretic difference $ \SATPOS(\alg{A})\setminus\SAT(\alg{A})$ 
of terms that are positively, but not fully satisfiable in $\alg{A}$ 
will be denoted shortly $\SATPOS\setminus\SAT(\alg{A})$. 

Clearly $\SAT(\stBool)=\SATPOS(\stBool)$. 
We say that an $\FLew$-algebra $\alg{A}$ has \emph{classical satisfiability} if $\SAT(\alg{A})=\SAT(\stBool)$;
 it has \emph{classical positive satisfiability} if $\SATPOS(\alg{A})=\SAT(\stBool)$.

\begin{lemma}
\label{sat_incl}
Let $\alg{A}$ be an $\FLew$-algebra; then $\SAT(\stBool)\subseteq\SAT(\alg{A})$.
 If moreover $\alg{A}$ is nontrivial, then 
$\SAT(\alg{A}) \subseteq \SATPOS(\alg{A})$.
\end{lemma}

Let us write $\oSAT$ and $\oSATPOS$ for complements of 
(the binary relations) $\SAT$ and $\SATPOS$;
thus $\oSAT(\alg{A})$ and $\oSATPOS(\alg{A})$ complement $\SAT(\alg{A})$ and $\SATPOS(\alg{A})$ in $\Fm$.

Let $\alg{A}$ be an $\FLew$-algebra and $\varphi(\bar x)$ an $\FLew$-term.
By definition of the $\SATPOS$ relation, 
$\oSATPOS(\alg{A}, \varphi)$ iff $\alg{A}\models\forall \bar x (\varphi(\bar x)\approx 0)$;
 we say that  $\varphi$ is \emph{unsatisfiable} in $\alg{A}$ (a \emph{contradiction}). 
The following lemma provides a relationship between the set of positively unsatisfiable terms
and \emph{negative tautologies}, i.e., tautologies in the form of a negated term.

\begin{lemma} \label{lmSATPOSvsTAUTNeg}
$\varphi\in \oSATPOS(\alg{A})$ iff $\neg\varphi \in \TAUT(\alg{A})$.
\end{lemma}

\begin{proof}
Use Fact \ref{FLew_facts}(1).
\end{proof}

$\SATPOS(\alg{A})$ and $\SAT(\alg{A})$ 
are syntactic fragments of the existential theory of $\alg{A}$;
indeed this is our definition of the two relations.
We point out that $\oSATPOS(\alg{A})$ and $\oSAT(\alg{A})$ 
are syntactic fragments of its universal theory.
Namely, $\oSATPOS(\alg{A})$ corresponds to the \emph{negative tautologies} of $\alg{A}$; 
algebraically, it is a syntactic fragment of the equational theory of $\alg{A}$.
Moreover, $\oSAT(\alg{A})$ is a syntactic fragment of the quasi-equational theory of $\alg{A}$:

\begin{lemma}
Let $\alg{A}$ be a nontrivial $\FLew$-algebra. Then $\varphi\in\oSAT(\alg{A})$
iff $\alg{A}\models \varphi\approx 1 \Rightarrow 0\approx 1$.
\end{lemma}

Let us look at the behaviour of some class operators with respect to satisfiability. Let $\alg{A}$, $\alg{B}$ be $\FLew$-algebras.
If $\alg{B}$ is a subalgebra of $\alg{A}$, then $\SAT(\alg{B})\subseteq \SAT(\alg{A})$ and $\SATPOS(\alg{B})\subseteq \SATPOS(\alg{A})$;
the inclusions can be strict, as exemplified by taking $\stBool$ as $\alg{B}$  and the standard MV-algebra as $\alg{A}$.
If $\alg{B}$ is a homomorphic image of $\alg{A}$, then $\SAT(\alg{A})\subseteq \SAT(\alg{B})$.
If $\alg{A}_i, i\in I$ is a family of $\FLew$-algebras and 
$\prod_i \alg{A}_i$ is the product of $\alg{A}_i, i\in I$, then $\SAT(\prod_i \alg{A}_i) = \bigcap_i \SAT(\alg{A}_i)$
and $\SATPOS(\prod_i \alg{A}_i) = \bigcup_i \SATPOS(\alg{A}_i)$.

\medskip

The following is Theorem 3.4.1 (ii) of \cite{Hanikova:Handbook}.
Corollary \ref{CF-NP-complete} in this paper provides
a simpler proof.

\begin{theorem} {\rm (\cite{Hanikova:Handbook})} 
\label{th_SAT_SATPOS_NP-hard} 
Let $\alg{A}$ be a nontrivial $\FLew$-algebra.
Then $\SAT(\alg{A})$ and $\SATPOS(\alg{A})$ are $\NP$-hard.
\end{theorem}

In Section \ref{section:generalOnSAT}, satisfiability and positive satisfiability 
are discussed with respect to their importance and related notions from logic.

\section{Glivenko equivalence and positive satisfiability}
\label{section:satpos}

The relation between 
positively satisfiable terms and negative tautologies of an $\FLew$-algebra,
given as Lemma \ref{lmSATPOSvsTAUTNeg}, is explored in this section. This connection
gives a clearer picture of various $\SATPOS$ problems for $\FLew$-algebras,
especially as regards their partial order by inclusion.

Further, a characterization is given of those $\FLew$-algebras $\alg{A}$
for which the set $\SATPOS(\alg{A})$ coincides with $\SAT(\stBool)$.
Any $\FLew$-algebra $\alg{A}$ with this property is nontrivial,
and by Lemma \ref{sat_incl}, for any such $\alg{A}$ 
also $\SAT(\alg{A})$ coincides with $\SAT(\stBool)$
and the set $\SATPOS\setminus\SAT(\alg{A})$ is empty.

One may consider $\Tm^\neg = \{\neg\varphi \mid \varphi \in \Tm\}$, and for
any logic $\logic{K}$ define
\begin{align*}
\logic{K}^\neg &= \logic{K} \cap \Tm^\neg\\
\TAUT^\neg(\alg{A}) &= \TAUT(\alg{A}) \cap \Tm^\neg
\end{align*}
By definition, $\neg\varphi\in \TAUT(\alg{A})$ iff $\neg\varphi\in \TAUT^\neg(\alg{A})$.

\begin{lemma}
\label{neg_subvalence}
Let $\alg{A}$, $\alg{B}$ be $\FLew$-algebras.
Then $$\SATPOS(\alg{A}) \subseteq \SATPOS(\alg{B})\mbox{ \ \ iff \ \ }\TAUT^\neg(\alg{B}) \subseteq \TAUT^\neg (\alg{A}).$$ 
\end{lemma}

\begin{proof} Using Lemma  \ref{lmSATPOSvsTAUTNeg},
$\oSATPOS(\alg{B})\subseteq\oSATPOS(\alg{A})$ iff $\TAUT^\neg(\alg{B}) \subseteq \TAUT^\neg (\alg{A})$.
\end{proof}

\begin{corollary} 
\label{neg_equivalence}
$\SATPOS(\alg{A}) = \SATPOS(\alg{B})$ iff  $\TAUT^\neg(\alg{A}) = \TAUT^\neg(\alg{B})$.
\end{corollary}

The equality $\TAUT^\neg(\alg{A}) = \TAUT^\neg(\alg{B})$ is an instance of an equivalence relation on logics 
(here, the logic of $\alg{A}$ and the logic of $\alg{B}$), 
well known as \emph{Glivenko equivalence} 
(\cite{Glivenko:Equivalence}; see also \cite{Galatos-JKO:ResiduatedLattices}, Chapter 8).
For logics extending $\FLew$, this notion can be slightly extended as follows.

\begin{definition}
Let the logics $\logic{K}$ and $\logic{L}$ extend $\FLew$. 
\begin{itemize}
\item $\logic{K}$ is \emph{Glivenko subvalent} to $\logic{L}$ 
(write $\logic{K} \subseteq^G \logic{L}$) 
iff $\logic{K}^\neg \subseteq \logic{L}^\neg$.
\item $\logic{K}$ is \emph{Glivenko equivalent} to $\logic{L}$ 
(write $\logic{K} \sim^G \logic{L}$) 
iff $\logic{K} \subseteq^G \logic{L}$ 
and $\logic{L} \subseteq^G \logic{K}$. 
\end{itemize}
\end{definition}

Quite a lot is known about Glivenko equivalence for substructural logics (within and beyond the realm of $\FLew$);
we take our references mainly from Chapter 8 of \cite{Galatos-JKO:ResiduatedLattices}.
Each equivalence class is convex, with a least and a greatest element in the lattice order of logics extending $\FLew$
(it forms a complete sublattice). 
Since Glivenko subvalence is a preorder, it yields naturally a partial order $\preceq$
on the classes of Glivenko equivalence.
In this order, the top element is the class containing (only) the inconsistent logic, the only coatom is
the class containing classical logic, and the bottom element is the class containing $\FLew$.

Let us point out that while product logic is incomparable to {\L}ukasiewicz logic,
 the class containing {\L}ukasiewicz logic
(and also H\'ajek's $\logic{BL}$, as shown in \cite{Cignoli-Torrens:HajekBasicFuzzyLogic}),
is subvalent to the class containing product logic (and classical logic, as explained below).

One can argue that there are continuum many 
Glivenko equivalence classes within the lattice of subvarieties of $\aFLew$
as follows.
Using \cite{Galatos-JKO:ResiduatedLattices}, Theorem 8.7, 
any involutive logic extending $\FLew$ is the largest element of its Glivenko equivalence class;
this means that distinct involutive logics belong to distinct Glivenko equivalence classes within $\FLew$,
and the cardinality of the set of involutive extensions of $\FLew$ is 
a lower bound on the number of Glivenko equivalence classes.
By \cite{Galatos-JKO:ResiduatedLattices}, Theorem 9.45, there is a continuum of involutive extensions of $\FLew$.
As the lattice of logics extending $\FLew$ has itself the cardinality of the continuum, 
there are continuum many classes of Glivenko equivalence within $\FLew$. 

Lemma \ref{neg_subvalence} carries the order $\preceq$ on Glivenko equivalence classes 
onto $\SATPOS$ problems for $\FLew$-algebras.
The smallest set $T$ of $\FLew$-terms such that 
$T=\SATPOS(\alg{A})$ for some $\FLew$-algebra $\alg{A}$
is the empty set, for $\alg{A}$ trivial. 
Its only direct successor in this order is $\SAT(\stBool)$.
The greatest element in this order is the set of terms $\{\varphi \mid \neg\varphi \not\in \FLew\}$.
The above cardinality argument for Glivenko equivalence classes yields the following statement.

\begin{lemma}
There are continuum many pairwise distinct sets $\SATPOS(\alg{A})$, for $\alg{A}$ an $\FLew$-algebra.
\end{lemma}

We further characterize the class of $\FLew$-algebras with a classical $\SATPOS$ problem,
relying on the Glivenko equivalence class of classical logic within $\FLew$.

\begin{theorem}
\label{not_WCon_SAT}
Let $\alg{A}$ be an $\FLew$-algebra and $\alg{A}\not\in\aWCon$. 
Then $\SATPOS\setminus\SAT(\alg{A})$ is nonempty.
\end{theorem}

\begin{proof}
If $\alg{A}$ is an $\FLew$-algebra but not a $\WCon$-algebra, then it is nontrivial and,
by assumption, the identity $x\wedge \neg x \approx 0$ is not valid in $\alg{A}$.
This means that, for some assignment $e_\alg{A}$ in $\alg{A}$, one has 
$e_\alg{A}(x\wedge \neg x) > 0^\alg{A}$. 
Hence the term $x\wedge \neg x$ is in  $\SATPOS(\alg{A})$.

On the other hand, using Fact \ref{FLew_facts} (1), the term $x\wedge \neg x$
is not in $\SAT(\alg{A})$ for any nontrivial $\FLew$-algebra $\alg{A}$.
Hence $x\wedge \neg x \in \SATPOS\setminus\SAT(\alg{A})$.
\end{proof}

In particular, $\SATPOS\setminus \SAT(\alg{A})$ is nonempty
for all involutive $\FLew$-algebras that are not Boolean algebras,
using Fact \ref{fact_inv_wcon}. Examples include the Chang algebra $\alg{K}_2$, which has classical SAT,
or $\mathrmL_3$, which has nonclassical SAT.

Below we derive the converse of Theorem \ref{not_WCon_SAT}: 
on any nontrivial $\WCon$-algebra $\alg{A}$, 
one has $\SAT(\alg{A}) = \SATPOS(\alg{A}) = \SAT(\stBool)$.
This follows from Glivenko theorem for $\WCon$ with respect to classical logic,
presented in \cite{Cignoli-Torrens:GlivenkoBCK} as Corollary 5.3.
Glivenko theorem (\cite{Glivenko:Equivalence}) provides a double-negation interpretation 
of classical logic in intuitionistic logic. 
The paper \cite{Cignoli-Torrens:GlivenkoBCK} points out that the same reasons that support the theorem 
for intuitionistic logic support it also for $\WCon$.

\begin{theorem}{\rm (Glivenko theorem for $\WCon$, \cite{Cignoli-Torrens:GlivenkoBCK})}
\label{glivenko_theorem_wcon}
For any $\FLew$-term $\varphi$,
one has $\varphi \in \logic{CL}$ iff $\neg\neg\varphi \in \WCon$.
\end{theorem}

It is easy to see that if a consistent $\logic{L}$ extends $\FLew$, then a Glivenko theorem holds for $\logic{L}$
with respect to classical logic iff $\logic{L}$ is Glivenko equivalent to classical logic:
let $\varphi$ be an $\FLew$-term.
For the left-to-right implication, any consistent $\FLew$-extension is Glivenko subvalent to classical logic,
and on the other hand $\neg\varphi \in \logic{CL}$ entails 
$\neg\neg\neg\varphi \in \logic{L}$ by assumption, hence 
$\neg\varphi\in\logic{L}$ by Fact \ref{FLew_facts} (3).
For the right-to-left implication, $\neg\neg\varphi\in \logic{L}$ clearly entails $\varphi\in \logic{CL}$,
and on the other hand $\varphi \in \logic{CL}$ gives $\neg(\neg \varphi) \in \logic{CL}$ and
$\neg(\neg \varphi) \in \logic{L}$ by assumption.

\begin{corollary}
\label{cr:Glivenko_equivalent}
$\WCon$ is Glivenko equivalent to $\CL$.
\end{corollary}

\begin{theorem}
Among consistent logics extending $\FLew$, $\WCon$ is the weakest logic that is Glivenko equivalent to classical logic.
\end{theorem}

\begin{proof}
By Corollary \ref{cr:Glivenko_equivalent}, 
$\WCon$ (and any of its consistent extensions) belongs to the Glivenko equivalence class of classical logic.
On the other hand, if $\logic{L}$ is Glivenko equivalent to classical logic,
then $\neg (x \wedge \neg x) \in \logic{L}$ (as it is a negative tautology of $\CL$). This implies that $\logic{L}$ extends WCon.
\end{proof}

Thus, while tautologousness for $\WCon$-algebras presents a rich structure 
(note that the lattice of superintuitionistic logics has the cardinality of the continuum),
the satisfiability and the positive satisfiability problem for these algebras 
are always identical to the classical satisfiability problem.  
This accounts for its computational complexity: it is $\NP$-complete.

Summing up, we have shown:

\begin{theorem} 
\label{theorem:char_satpos}
Let $\alg{A}$ be a nontrivial\/ $\FLew$-algebra. The following are equivalent:
\begin{itemize}
\item[(1)] $\alg{A}$ is a $\WCon$-algebra;
\item[(2)] $x \wedge \neg x \in \oSATPOS(\alg{A})$;
\item[(3)] $\SATPOS(\alg{A}) = \SAT(\stBool)$;
\item[(4)] $\SATPOS(\alg{A}) =  \SAT(\alg{A})$. 
\end{itemize}
\end{theorem}

\section{Full satisfiability}
\label{section:fulsat}

This section focuses on characterizing classical full satisfiability in a $\FLew$-algebra
by the existence of a two-element congruence in that algebra (or equivalently, 
of a homomorphism onto $\stBool$). 
Moreover, a classically unsatisfiable term is proposed that occurs
in $\SAT(\alg{A})$ for every $\FLew$-chain $\alg{A}$ such that $\SAT(\alg{A})$ is not classical.
We also show that the $\SAT$ problem for the standard MV-algebra contains
the $\SAT$ problem of each nontrivial BL-algebra. 

The following theorem shows that in order to find out whether an $\FLew$-chain 
(i.e., an $\MTL$-chain) has classical satisfiability, it is enough to know whether 
it satisfies a single term in one variable.

\begin{theorem} \label{thChainsClassicalSAT}
For a nontrivial $\FLew$-chain $\alg{A}$, the following are
equivalent:
\begin{itemize}
\item[(1)] $\SAT(\alg{A})$ is classical;
\item[(2)] $(\neg x \to x) \wedge \neg (x^3) \in \oSAT(\alg{A})$;
\item[(3)] $\neg$ in $\alg{A}$ has no fixed point and the
      set $\{x \in A \mid x^2 > 0^\alg{A} \}$ is closed under multiplication;
\item[(4)] there is a homomorphism $h: \alg{A} \rightarrow \stBool$.\footnote{Such a homomorphism is surjective, as it preserves the two constants.}
\end{itemize}
\end{theorem}

\begin{proof}
Clearly, (4) implies (1) and (1) implies (2).
Let us prove (2) implies (3) by proving that if (3) is not
satisfied, then $(\neg x \to x) \wedge \neg (x^3) \in \SAT(\alg{A})$.

Let $\phi(x) = (\neg x \to x) \wedge \neg (x^3)$.
If the negation has a fixed point $a$,
then $(\neg a \to a) = 1^\alg{A}$ and $a^2 = 0^\alg{A}$.
Hence, we have $\phi(a)=1^\alg{A}$ and $\phi \in \SAT(\alg{A})$.
If the negation has no fixed point and the set
$\{x \mid x^2 > 0^\alg{A} \}$ is not closed under multiplication,
then there are $a, b \in A$ such that $a \le b$,
$a^2 > 0^\alg{A}$, $b^2 > 0^\alg{A}$, and $(a \cdot b)^2 = 0^\alg{A}$.
This implies $a^4 = 0^\alg{A}$ and $\neg a < a$.
If $a^3=0^\alg{A}$, we have $\phi(a)=1^\alg{A}$ and $\phi \in \SAT(A)$.
If $a^3 > 0^\alg{A}$, then $\neg a \ge a^2$ is not satisfied and
also $\neg (a^2) \ge a$ is not satisfied. Since $A$ is a chain,
we have $\neg a < a^2$ and $\neg (a^2) < a$.
Moreover, since $a^4 = 0^\alg{A}$, we have $a^2 \le \neg (a^2)$.
Hence,
$$
\neg a < a^2 \le \neg (a^2) < a
$$
and, clearly,
$$
\neg a < a^2 \le \neg \neg (a^2) \le \neg (a^2) < a \ .
$$
Let $c = \neg (a^2)$. Since $c \le a$, we have
$c^2 \le a^2 \le \neg \neg (a^2) = \neg c \le c$.
In particular, $c^2 \le \neg c \le c$,
which implies $\phi(c)=1^\alg{A}$ and $\phi \in \SAT(A)$.

It remains to prove (3) implies (4). Assume an $\FLew$-chain
$\alg{A}$ satisfying (3) is given.
Define
\begin{align*}
A_0 &= \{x \mid x^2 = 0^\alg{A} \},\\
A_1 &= \{x \mid x^2 > 0^\alg{A} \}.
\end{align*}
Note that $0^\alg{A} \in A_0$ and $1^\alg{A} \in A_1$, so the sets
are both nonempty. Clearly, the sets are disjoint and 
$A = A_0 \cup A_1$. Define a map $h$ by $h(A_0) = 0$ and $h(A_1) = 1$
and let us verify that $h$ is a homomorphism of $\alg{A}$ to $\stBool$.

The set $A_1$ is an upper set. Moreover, (3) implies that $A_1$
is a congruence filter and, hence,
a class of a congruence on $\alg{A}$ containing $1^\alg{A}$.
Let $a \in A_0$. Clearly, $a \le \neg a$ and $\neg (\neg a) \le \neg a$.
Since $\neg$
has no fixed point in $A$, we have $(\neg a)^2 > 0^\alg{A}$, so $\neg a$
is congruent to $1^\alg{A}$. Hence, $a$ and $0=a \cdot \neg a$ are
in the same class of the congruence.
It follows that all elements of $A_0$ are
in the same class, so $A_0$ is a class of the congruence.
Since $A_0$ and $A_1$ are classes of a congruence, $h$ is
a homomorphism.
\end{proof}

Some $\FLew$-algebras can be decomposed into two congruence classes:
for example, the Chang algebra $\alg{K}_2$ can be decomposed into infinitesimals and
co-infinitesimals. There are also non-chains with this property, 
for example, the product $K_2 \times \alg{A}$ or $\stBool \times \alg{A}$
for a nontrivial $\FLew$-algebra $\alg{A}$. 
Any such algebra has classical satisfiability,
as the homomorphism onto $\stBool$ preserves satisfiability of terms.
In Corollary \ref{corSATvsHom} below, we prove that
this property \emph{characterizes} $\FLew$-algebras with classical satisfiability.
We first address finitely generated $\FLew$-algebras.

\begin{theorem} \label{thLocalHomomorphism}
Let $\alg{A}$ be a nontrivial finitely generated $\FLew$-algebra
with generators $a_1, \ldots, a_k \in A$.
The following are equivalent:
\begin{itemize}
\item[(1)] $\SAT(\alg{A})$ is classical;
\item[(2)] for every $\FLew$-term $\phi(x_1,...,x_k)$, if
           $\phi(a_1,...,a_k)=1^\alg{A}$, then $\phi \in \SAT(\stBool)$;
\item[(3)] there is a homomorphism $h: \alg{A} \rightarrow \stBool$.
\end{itemize}
\end{theorem}

\begin{proof}
Clearly, (3) implies (1) and (1) implies (2). Let us prove
that (2) implies (3).

Assume (2).
First, we prove by contradiction that there are elements
$b_1,...,b_k \in \{0^\alg{A},1^\alg{A}\}$, such that for every
term $\phi(x_1,...,x_k)$, if $\phi(a_1,...,a_k)=1^\alg{A}$, then
$\phi(b_1,...,b_k)=1^\alg{A}$. If not, then for each of the $2^k$
possible choices of
$b_1,...,b_k \in \{0^\alg{A},1^\alg{A}\}$, there is a term
$\phi(x_1,...,x_k)$ such that $\phi(a_1,...,a_k)=1^\alg{A}$
and $\phi(b_1,...,b_k)=0^\alg{A}$. The conjunction of the $2^k$
terms obtained in this way is a classically
unsatisfiable term satisfied by $a_1,...,a_k$ in $\alg{A}$. This is a contradiction
to (2). Hence, the elements $b_1,...,b_k$ with the required
property exist.

Let $b_1,...,b_k \in \{0^\alg{A},1^\alg{A}\}$ be the elements guaranteed
by the previous paragraph. Let us prove that there is a 
homomorphism $h:A \to \{0^\alg{A},1^\alg{A}\}$, such that for $i=1,\ldots,k$,
we have $h(a_i)=b_i$.
Since $a_1,...,a_k$ generate $\alg{A}$,
there is at most one such homomorphism. Let $U$ be the
closure of the elements $(a_i,b_i)$ in $A \times \{0^\alg{A},1^\alg{A}\}$. There
is a homomorphism extending $h(a_i)=b_i$ for $i=1,\ldots,k$ 
iff the relation $U$ is a function.
Assume for a contradiction that $U$ is not a function.
Then there is some $c \in A$ such that $(c,0^\alg{A}) \in U$ and
$(c,1^\alg{A}) \in U$. Since $(c,0^\alg{A}) \in U$, there is a term
$\phi_0$ such that
\begin{eqnarray*}
\phi_0(a_1,...,a_k) & = & c \\
\phi_0(b_1,...,b_k) & = & 0^\alg{A} \ .
\end{eqnarray*}
Similarly, since $(c,1^\alg{A}) \in U$, there is a term
$\phi_1$ such that
\begin{eqnarray*}
\phi_1(a_1,...,a_k) & = & c \\
\phi_1(b_1,...,b_k) & = & 1^\alg{A} \ .
\end{eqnarray*}
Let $\phi = \phi_1 \to \phi_0$. The term $\phi$ satisfies
$\phi(a_1,...,a_k) = 1^\alg{A}$ and $\phi(b_1,...,b_k) = 0^\alg{A}$. Since
this is a contradiction with the construction of $b_1,...,b_k$,
we can conclude that $U$ is a function and defines
a homomorphism of $\alg{A}$ to $\{0^\alg{A},1^\alg{A}\}$. 
\end{proof}

Finitely generated subalgebras of a general $\FLew$-algebra
$\alg{A}$ suffice to determine whether $\SAT(\alg{A})$
is classical, in the following sense.

\begin{lemma} \label{lmExtensionSAT}
An $\FLew$-algebra $\alg{A}$ has classical
satisfiability iff each of its finitely generated
subalgebras has classical satisfiability.
\end{lemma}

\begin{proof}
If $\SAT(\alg{A})$ is classical, then this is true also
for all subalgebras of $\alg{A}$.
If $\SAT(\alg{A})$ contains a term in $k$ variables that
is not classically satisfiable, then this term is satisfiable
in some subalgebra of $\alg{A}$ with at most $k$ generators.
\end{proof}

\begin{theorem} \label{thExtensionHom}
Let $\alg{A}$ be an $\FLew$-algebra. There is
a homomorphism $h: \alg{A} \rightarrow \stBool$
iff a homomorphism
$h: \alg{B} \rightarrow \stBool$ exists for every
finitely generated subalgebra $\alg{B}$ of $\alg{A}$.
\end{theorem}

\begin{proof}
If there is a homomorphism $h: \alg{A} \rightarrow \stBool$,
then its restriction to any finitely generated subalgebra
is again a homomorphism.

For the opposite direction, let $G$ be the set of all finite
subsets of $A$ and let $\alg{B}_g$, $g \in G$ be the subalgebra of $\alg{A}$
generated by $g$. 
Assume that for each $g \in G$ there
is a homomorphism $h_g:\alg{B}_g \to \stBool$.
It is well known that $\alg{A}$ is embeddable into 
an ultraproduct ${\cal P}$ of $\alg{B}_g$, $g \in G$, given by an ultrafilter ${\cal F}$ on $G$,
via an embedding $\mu:\alg{A} \to {\cal P}$.
Define a mapping $h:{\cal P} \to \stBool$:
for each $u \in {\cal P}$ and each $f:G \to A$ such that $u=[f]_{\cal F}$, 
the value $h(u)$ is chosen so that $\{g \mid h_g(f(g)) = h(u) \}$ belongs to ${\cal F}$. 
One can easily verify that the definition of $h(u)$ is correct and that $h$ is a homomorphism.
The composition of $\mu$ and $h$ yields
a homomorphism from $\alg{A}$ onto $\stBool$.
\end{proof}

\begin{corollary} \label{corSATvsHom}
Let $\alg{A}$ be an $\FLew$-algebra. Then $\SAT(\alg{A})$ is classical
iff there is a homomorphism from $\alg{A}$ onto $\stBool$. 
\end{corollary}

\begin{proof}
The right-to-left implication is trivial. For the left-to-right
implication, use Lemma \ref{lmExtensionSAT},
Theorem \ref{thLocalHomomorphism} and Theorem \ref{thExtensionHom}.
\end{proof}

The maximal set of satisfiable terms for $\BL$-algebras can be characterized
relying on the ordinal sum representation of $\BL$-chains.\footnote{Cf.~\cite{Hajek:1998, Mostert-Shields:OrdinalSum}.
The theorem says that each saturated BL-chain is an ordinal sum of MV-components, $G$-components and $\Pi$-components, 
and in particular, either there is a first MV-component or the chain is an $\SBL$-algebra.}

\begin{theorem}
Let $\alg{A}$ be a nontrivial $\BL$-algebra. Then $\SAT(\alg{A}) \subseteq \SAT(\standardL)$.
\end{theorem}

\begin{proof}
Each $\BL$-algebra $\alg{A}$ is a subdirect product of $\BL$-chains, say $\Pi_i \alg{A}_i$.
Consider an $\FLew$-term $\varphi$. Its satisfiability in $\alg{A}$ entails satisfiability in $\Pi_i \alg{A}_i$
and a fortiori in each $\alg{A}_i$. 
For each $i$, either $\alg{A}_i$ is an $\SBL$-chain, or its saturation has a first MV-component. 
If $\alg{A}_i$ is an $\SBL$-chain for each $i$, then satisfiability (in each $\alg{A}_i$ and)
in $\alg{A}$ is classical, which yields the statement.
Assume for some $i$, $\alg{A}_i$ has a first MV-component $\alg{B}$ on $[0,b]$. 
Then $\alg{B}'=[0,b)\cup 1$ is a homomorphic image of $\alg{A}$ under the map $x \mapsto \neg\neg x$. 
Therefore, $\varphi$ is satisfiable in $\alg{B}'$. Since $\alg{B}'$ is partially 
embeddable in $\standardL$, $\varphi$ is satisfiable in $\standardL$.
Therefore, $\SAT(\alg{A})\subseteq \SAT(\standardL)$.
\end{proof}

$\SAT$ problems for $\FLew$-algebras are ordered by inclusion; $\SAT(\stBool)$ is the bottom element in this order.
The above theorem shows that for satisfiability problems for $\BL$-algebras, $\SAT(\standardL)$ is the top. 
We do not know whether there is a top in this order over all $\FLew$-algebras.

Consider two $\FLew$-algebras $\alg{A}$ and $\alg{B}$.
The previous section tells us that if $\SATPOS(\alg{A})=\SATPOS(\alg{B})=\SATPOS(\stBool)$,
then also $\SAT(\alg{A}) = \SAT(\alg{B})$. However, without the assumption that $\SATPOS$ 
is classical for both algebras, the implication does not hold; we shall see in the following
section that, in the realm of MV-algebras, one can obtain a continuum of distinct $\SAT$ problems,
while there are infinitely countably many $\SATPOS$ problems.

\section{Satisfiability in MV-algebras}
\label{section:MV}

Unlike $\WCon$-algebras discussed earlier in this paper,
MV-algebras (and MV-chains in particular) present a rich variety 
of distinct satisfiability problems.

Let us look first at the $\SATPOS$ relation for MV-algebras and, 
indeed, all involutive $\FLew$-algebras.
In an involutive $\FLew$-algebra $\alg{A}$, 
the function $\neg^\alg{A}$ is a bijection on ${A}$; 
in particular, the preimage of $0^\alg{A}$ is (just) $1^\alg{A}$.
In addition to Lemma \ref{lmSATPOSvsTAUTNeg} which holds for each $\FLew$-algebra,
for an involutive algebra $\alg{A}$ one also gets
$$ \varphi\in \TAUT(\alg{A}) \Leftrightarrow \neg\varphi\in \oSATPOS(\alg{A}) $$
by combining Lemma \ref{lmSATPOSvsTAUTNeg} for the instance $\neg\varphi$ 
with the equivalence $\varphi\equiv\neg\neg\varphi$.
In plain words, one can read tautologousness from positive satisfiability.
Moreover, for each term $\varphi$ and each involutive $\FLew$-algebra $\alg{A}$, one has 
$\varphi\in\TAUT(\alg{A})$ iff $\neg(\neg\varphi)\in\TAUT(\alg{A})$ iff $\neg(\neg\varphi)\in\TAUT^\neg(\alg{A})$.
Hence, for involutive $\FLew$-algebras $\alg{A}$ and $\alg{B}$, we have 
$\TAUT^{\neg}(A) = \TAUT^{\neg}(B) \Leftrightarrow \TAUT(A) = \TAUT(B)$.
Recalling Corollary \ref{neg_equivalence}, saying that for any two 
$\FLew$-algebras $\alg{A}, \alg{B}$ one has
$\SATPOS(\alg{A})=\SATPOS(\alg{B})$ iff  $\TAUT^\neg(\alg{A}) = \TAUT^\neg(\alg{B})$, 
we may conclude:

\begin{theorem}
For any two involutive $\FLew$-algebras $\alg{A}$ and $\alg{B}$, one has
$$\SATPOS(\alg{A})=\SATPOS(\alg{B}) \mbox{\ \ iff\ \ } \TAUT(\alg{A}) = \TAUT(\alg{B}).$$
\end{theorem}

Komori (\cite{Komori:SuperLukasiewiczPropositional}) provided a classification of subvarieties of $\Alg{MV}$
(hence, also, of the sets $\TAUT(\alg{A})$ for $\alg{A}$ being an MV-algebra):
there are countably infinitely many such subvarieties, 
and each is generated by a particular choice of algebras from among $\alg{K}_n$ and $\mathrm{\L}_n$, 
standing for the $n$-segment Komori algebra and the $n$-element finite MV-chain respectively.
Moreover, the generated variety is the full variety of MV-algebras iff the set of generators is infinite.
The cardinality of the set of problems $\SATPOS(\alg{A})$ for $\alg{A}$ an MV-algebra is therefore countably infinite.
\cite{Cintula-Hajek:ComplexityLukasiewicz} shows that axiomatic extensions of {\L}ukasiewicz logic 
(and therefore, all $\TAUT(\alg{A})$ problems for $\alg{A}$ being an MV-algebra) are $\coNP$-complete. 
It follows that, for any choice $\alg{A}$ of a nontrivial MV-algebra, the problem $\SATPOS(\alg{A})$ is $\NP$-complete.

We now turn to the SAT relation for MV-algebras. 
For MV-chains, we show that equality of $\SAT$ problems can
replace the requirement of containing the same rationals
in a characterization of equality of universal theories,
given in \cite{Gispert-Mundici:MVAlgebras}.
For an integer $n\geq 1$, one says that the rationals 
$\{0/n,1/n, 2/n, \dots, n/n\}$ are contained in an MV-chain $\alg{A}$ iff 
$\alg{A}$ contains an isomorphic copy of the $(n+1)$-element MV-chain $\mathrmL_{n+1}$
as a subalgebra. For the latter, we say shortly that $\alg{A}$ contains $\mathrmL_{n+1}$.
See \cite{Gispert-Mundici:MVAlgebras, DiNola-Leustean:Handbook} for the definitions of the notions of order and rank of an MV-chain, used below.

\begin{theorem}[\cite{Gispert-Mundici:MVAlgebras}, Theorem 6.7] 
Two MV-chains have the same universal theory iff they have the same order,
the same rank, and they contain the same rationals.
\end{theorem} 

Our aim is to rephrase this theorem in Corollary \ref{corSameTAUTAndSAT} in terms of $\TAUT$ and $\SAT$ operators. 
For this purpose, recall the result of \cite{Komori:SuperLukasiewiczPropositional} saying that two MV-chains have the
same $\TAUT$ problem iff they have the same order and the same rank.

\begin{lemma}
\label{rat_l1}
Let $\alg{A}$ be an MV-chain. Let $n$ be an integer, $n\geq 2$. 
Then $x \equiv (\neg x)^{n-1} \in \SAT(\alg{A})$ iff $\alg{A}$ contains 
$\mathrmL_{n+1}$.
\end{lemma}

\begin{proof}
If $\alg{A}$ contains $\mathrmL_{n+1}$, the element $1/n$ of $\mathrmL_{n+1}$ 
satisfies the term in $\alg{A}$. On the other hand, assume $x\equiv(\neg x)^{n-1}$ has a
solution $a$ in $\alg{A}$. Hence, $a=(\neg a)^{n-1}$ and
$\neg a = (n-1)a$. Clearly $0^\alg{A}<a<1^\alg{A}$. 
If $n=2$, the equation gives $a=\neg a$, so $\{0^\alg{A}, a, 1^\alg{A}\}$
is isomorphic to $\mathrmL_3$.
For the rest of the proof, we assume $n\geq 3$ and 
on this assumption, $0^\alg{A}<a \leq (\neg a)^2<\neg a<1^\alg{A}$.
We have $(\neg a)^n = (\neg a)^{n-1}\cdot \neg a = a\cdot\neg a = 0^\alg{A}$.  
Moreover $\neg a = (n-1)a$, so $n a=1^\alg{A}$.
In the rest, we show $\neg ka=(n-k)a$, by induction on $k$ from $1$ up to $n-1$.
For $k=1$ the statement holds by assumption.
Further we prove for $k\geq 2$ on the induction assumption for $k-1$,
i.e., $\neg(k-1)a = (n-k+1)a$.
We have $\neg ka = \neg (a+ (k-1) a) = \neg a \cdot \neg (k-1)a = \neg a \cdot (n-k+1)a = 
\neg a \cdot (a + (n-k)a)$.
The latter is equal to 
$(n-k) a \cdot (\neg (n-k)a + \neg a)$ and to $ (n-k)a \cdot \neg ((n-k)a\cdot a)$.
Since $(n-1)a \cdot a=0^\alg{A}$, we have $(n-k)a\cdot a=0^\alg{A}$.
Altogether, $\neg ka = (n-k)a\cdot \neg 0^\alg{A} = (n-k)a$.
It follows that the subset $\{0^\alg{A},a,2a,\dots, (n-1)a, 1^\alg{A}\}$ of $\alg{A}$
with the corresponding restrictions of the operations is isomorphic
to $\mathrmL_{n+1}$.
\end{proof}

\begin{lemma} \label{lmSameSAT}
If two MV-chains have the same $\SAT$ problem, they contain the same rationals.
\end{lemma}

\begin{proof}
By contraposition, if MV-chains $\alg{A}$ and $\alg{B}$ do not have the same rationals, then
there is an integer $n\geq 2$ such that, without loss of generality, $\alg{A}$ contains $\mathrmL_{n+1}$ while $\alg{B}$ does not.
By Lemma \ref{rat_l1},  $ x\equiv(\neg x)^{n-1}$ is in $\SAT(\alg{A})$ but not in $\SAT(\alg{B})$. 
Thus $\alg{A}$ and $\alg{B}$ do not have the same $\SAT$. 
\end{proof}

The following is a reformulation of Theorem 6.7 of \cite{Gispert-Mundici:MVAlgebras}
as described above.

\begin{corollary} \label{corSameTAUTAndSAT}
Two MV-chains have the same universal theory iff they have the same $\TAUT$ and the same $\SAT$ problems.
\end{corollary}

\begin{proof}
If two MV-chains have the same universal theory, they have the same $\TAUT$ problem and also 
the same $\SAT$ problem. For the opposite direction,
if two MV-chains have the same TAUT problem, they have the same order and the same rank. If they,
moreover, have the same SAT problem, they have the same rationals by Lemma \ref{lmSameSAT}. Hence,
by Theorem 6.7 of \cite{Gispert-Mundici:MVAlgebras}, they have the same universal theory. 
\end{proof}

We address a decidability issue for $\SAT$ now, 
offering a cardinality argument showing that the $\SAT$ problem is undecidable (indeed, nonarithmetical) 
for many subalgebras of the standard MV-algebra $\standardL$.

\begin{theorem}
\label{theorem:MV_SATs}
There are continuum many distinct problems $\SAT(\alg{A})$ for different choices of an MV-algebra $\alg{A}$.
\end{theorem}

\begin{proof}
Consider a class of subalgebras of $\standardL$ defined as follows.
Denote $P$ the set of primes. 
For $R\subseteq P$, denote $R^\ast$ the subalgebra of $\standardL$ generated by the set $\{ 1/r \mid r\in R\}$. 
Then clearly $R^\ast$ is a subalgebra of $\standardL \cap \mathds{Q}$, where $\mathds{Q}$ denotes 
the set of rational numbers.

Let $R\subseteq P$. Let $r/s\in R^\ast$, where $r$ and $s$ are coprime. 
Then $s$ is a product of primes from $R$ (each of multiplicity at most $1$).
For a $p\in P$, in particular, $R^\ast$ contains $\mathrmL_{p+1}$  iff $p\in R$, and
using Lemma \ref{rat_l1}, the term $ x\equiv (\neg x)^{p-1}$ is satisfiable in $R^\ast$ iff $p\in R$.
For every $R_1, R_2 \subseteq P$, this implies
$\SAT(R_1^\ast)\subseteq \SAT(R_2^\ast)$ iff $R_1\subseteq R_2$.
To conclude, consider that there are continuum many sets of primes pairwise incomparable by inclusion.
\end{proof}

The above proof can probably be considered folklore for MV-algebras.
Note that $\SAT(R^\ast)$ is undecidable whenever $R$ is, but we do not know whether the converse is true,
i.e., whether a decidable $R$ gives rise to a decidable $\SAT(R^\ast)$.
Although the result speaks of MV-algebras, it pertains to $\SAT$ problems in $\FLew$-algebras in general. 
When pondering the complexity of all possible $\SAT$ problems for $\FLew$-algebras,
one can make, on cardinality alone, the following conclusion.

\begin{corollary} \label{cor_SAT_undecidable}
A majority of $\SAT$ problems for $\FLew$-algebras are nonarithmetical.
\end{corollary}

\section{Some syntactic fragments} 
\label{section:CNFandDNF}

We define two particular types of {\FLew} terms:
terms in \emph{conjunctive form} and terms in \emph{disjunctive form},
since satisfiability within these fragments in
a general $\FLew$-algebra is the same as in $\stBool$.
Although both notions are akin to CNF's and DNF's of classical logic
in this sense, their properties here are rather different, and in particular,
neither represents all definable functions in a general $\FLew$-algebra.

\begin{definition}{\ } 
\begin{itemize}
\item A \emph{literal} is a variable (such as $x$), or a negation
      thereof (such as $\neg x$).
\item A \emph{$(\cdot, \vee)$-term} is any term built up from literals using
      an arbitrary combination of the symbols $\cdot$ and $\vee$.
\item In particular, a \emph{clause}/\emph{monomial} 
      is a term built up from literals using only the symbols $\vee$/only the symbols $\cdot$.
\item A term is in \emph{conjunctive form} (a CF-term),
      iff it is built up from clauses using only the symbols $\cdot$.
\item A term is in \emph{disjunctive form} (a DF-term),
      iff it is built up from monomials using only the symbols $\vee$. 
\end{itemize}
\end{definition}

The above definition uses the multiplication $\cdot$ rather
than the lattice meet $\wedge$ in rendering the conjunction of classical logic.
The reason is that $\cdot$ distributes over $\vee$ in $\FLew$
(cf.~Fact \ref{dist_law}), while $\wedge$ in general does not,
thus $\cdot$ better approximates a key interaction between the classical operations.
Other classical properties that our function symbols ($\cdot$, $\vee$, and also $\wedge$) 
retain are  the laws of commutativity and associativity; 
$\vee$ and $\wedge$ are moreover idempotent.
That is why one might consider $\vee$-disjunctions as sets of literals.
On the other hand, the multiplication $\cdot$ is not
in general idempotent; for $\cdot$, one must dispense with the (classically implicit)
assumption that the arguments of a conjunction may be specified as a set.

Using distributivity of $\cdot$ over $\vee$, 
one can bring any CF-term to a DF-term. Not every $\FLew$-term has an equivalent CF- or DF-term.
For example, for a free Heyting algebra with $n$ generators, there are infinitely
many non-equivalent terms in $n$ variables. On the other hand, since $\cdot$ is
idempotent in Heyting algebras, there are only finitely many non-equivalent DF-terms
of any given number of variables. 
Another example is the term $x \wedge y$: 
consider the standard product algebra $\standardP$ and choose $a,b \in [0,1]$, such
that $b^2 < a < b$. If $\phi(x, y)$ is a DF-term, then one can
show\footnote{$\phi(x, y)$ has to contain
a monomial, which evaluates to $a$ with $x=a$ and $y=b$, and the only monomial
with this property is $x$.}
that $\phi(a, b) = a$ implies $\phi(x, y) \ge x$. Hence, $\phi(x, y)$ does not
define $x \wedge y$.

\begin{lemma}
\label{lm_SAT_Bool}
Let $\alg{A}$ be a nontrivial $\FLew$-algebra. 
A $(\cdot, \vee)$-term is (positively) satisfiable in $\alg{A}$
iff it is satisfiable in $\stBool$.
\end{lemma}

\begin{proof} Let $\varphi$ be a $(\cdot, \vee)$-term. 
One can bring $\varphi$ to an $\FLew$-equivalent DF-term $\varphi^\circ$ using distributivity 
of  $\cdot$ over $\vee$  (Fact \ref{dist_law}).  
If $\varphi^\circ$ is satisfiable in $\stBool$, 
then both $\varphi^\circ$ and  $\varphi$ are satisfiable in $\alg{A}$ 
by the same assignment.
On the other hand, if $\varphi^\circ$ is unsatisfiable in $\stBool$,
then each of its monomials contains a pair of complementary literals, i.e., 
(relying on commutativity of $\cdot$) a subterm of the form $x\cdot\neg x$ for some variable $x$.
The term $x\cdot\neg x$ is $\FLew$-equivalent to $0$.
Therefore, each of the monomials in $\varphi^\circ$ is interpreted as  $0^\alg{A}$.
Hence, $\varphi^\circ$ is (positively) unsatisfiable in $\alg{A}$, and so is $\varphi$.
 \end{proof}

The restriction to CF-terms of the set $\SAT(\alg{A})$ will be denoted $\SAT^{{\rm CF}}(\alg{A})$.
The sets $\SATPOS^{{\rm CF}}(\alg{A})$, $\oSAT^{{\rm CF}}(\alg{A})$ and $\oSATPOS^{{\rm CF}}(\alg{A})$
are defined analogously.

\begin{corollary} 
\label{CF-NP-complete}
For any nontrivial $\FLew$-algebra $\alg{A}$, 
the sets $\SAT^{\rm CF}(\alg{A})$ and $\SATPOS^{\rm CF}(\alg{A})$
are $\NP$-complete.
\end{corollary}

\begin{proof}
Lemma \ref{lm_SAT_Bool} gives
$$
\SAT^{\rm CF}(\stBool) = \SAT^{\rm CF}(\alg{A}) = \SATPOS^{\rm CF}(\alg{A})$$
for any nontrivial $\FLew$-algebra $\alg{A}$.
\end{proof}

Let us mention the problem of computing or approximating
the maximum value of a term
in the standard MV-algebra $\standardL$. For a term $\phi$
with variables $x_1,\ldots,x_n$, let $\max(\phi)$ be the
maximum of its value over $[0,1]^n$, which is well defined,
since it is the maximum of a continuous function on a compact
subset of the Euclidean space. The following theorem may
be obtained as a consequence of Lemma \ref{lm_SAT_Bool},
however, we present a proof based on geometric properties
of $(\cdot, \vee)$-terms in $\standardL$.

\begin{theorem}
If $\delta < 1/2$ is a positive real constant and there
is a polynomial-time algorithm computing, for every $(\cdot, \vee)$-term
$\phi$ in the algebra $\standardL$, a real number $\mathrm{alg}(\phi)$
satisfying $|\mathrm{alg}(\phi) - \max(\phi)| \le \delta$,
then $\Po=\NP$.
\end{theorem}

\begin{proof}
The interpretations of the literals $x_i$ and $\neg x_i$ are
linear functions and, hence, they are convex in the geometric sense.
Moreover, if $f,g$ are convex functions, then so are $\max(f,g)$
and $\max(0, f+g-1)$ taken pointwise. Using simple induction,
the interpretation of any $(\cdot, \vee)$-term $\varphi(x_1,\dots,x_n)$
in $\standardL$ is a convex function in $[0,1]^n$.
Since $[0,1]^n$ is the convex hull of $\{0,1\}^n$, the maximum
of $\varphi(x_1,\dots,x_n)$ over $[0,1]^n$ is equal to
its maximum over $\{0,1\}^n$.

Assume an algorithm exists with the property given in the theorem.
Since $\delta < 1/2$ and for every $(\cdot, \vee)$-term,
we have $\max(\phi) \in \{0,1\}$, we have also
$$
\phi \in \SAT(\stBool) \iff \max(\phi)=1 \iff \mathrm{alg}(\phi) \ge 1/2 \ .
$$
Since testing the leftmost condition is $\NP$-hard
and the rightmost condition can be verified
using the output $\mathrm{alg}(\phi)$ of the algorithm, the
conclusion follows.
\end{proof}

\section{Positive, but not full satisfiability}
\label{section:DP}

For each nontrivial $\FLew$-algebra $\alg{A}$,
we investigate the set $\SATPOS \setminus \SAT(\alg{A})$ 
and we look at this set from a computational point of view.
Recall that both $\SAT(\alg{A})$ and $\SATPOS(\alg{A})$
are $\NP$-hard for a nontrivial $\FLew$-algebra $\alg{A}$
(Theorem \ref{th_SAT_SATPOS_NP-hard}).

Suppose that both $\SAT(\alg{A})$ and $\SATPOS(\alg{A})$ are $\NP$-sets.
Then it follows from the definition that the set $\SATPOS \setminus \SAT(\alg{A})$ 
 is a $\Delta_2$ set within the polynomial hierarchy. 
 The class of decision problems that arise as a set-theoretic difference of two problems in $\NP$ 
has been shown to have complete problems under polynomial-time reducibility.

\begin{definition}
A decision problem $L$ is in the class $\DP$ iff $L=L_1\setminus L_2$ for some decision problems $L_1, L_2\in\NP$.
\end{definition}

\begin{fact}[\cite{Papadimitriou:CC}]
\label{th_standard_DP_complete}
If $L_1, L_2$ are $\NP$-complete sets, then $L_1 \times \overline{L_2}$
is $\DP$-complete.
\end{fact}

\noindent
Examples of algebras $\alg{A}$ with $\SATPOS \setminus \SAT(\alg{A}) \in \DP$ include the standard
and the finite MV-algebras, since the corresponding $\SATPOS(\alg{A})$ and
$\SAT(\alg{A})$ are in NP.
$\SAT^{{\rm CF}}(\stBool)$ is $\NP$-complete.
Consequently,
\begin{equation} \label{fl_DP_complete}
\SAT^{{\rm CF}}(\stBool) \times \oSAT^{{\rm CF}} (\stBool)
\end{equation}
is complete for $\DP$.
One may assume that, for each given instance $\langle \varphi_1, \varphi_2 \rangle$, 
the terms $\varphi_1$ and $\varphi_2$ share no variables.

We show the following as a lower bound 
on the complexity of $\SATPOS\setminus\SAT(\alg{A})$.

\begin{theorem} \label{th_DP_hard}
 Let $\alg{A}$ be an $\FLew$-algebra. 
Assume that $\SATPOS \setminus \SAT(\alg{A})$ is  nonempty.
Then $\SATPOS \setminus \SAT(\alg{A})$ is  $\DP$-hard.
\end{theorem}

\begin{proof}
The assumptions on $\alg{A}$ entail nontriviality. 
Let $\alpha\in\SATPOS \setminus \SAT(\alg{A})$.
We give a polynomial-time reduction of the set (\ref{fl_DP_complete})
to $\SATPOS \setminus \SAT(\alg{A})$.
For any pair of CF-terms $\langle \varphi_1,\varphi_2\rangle$,
such that $\alpha$, $ \varphi_1$ and $\varphi_2$ have no variables
in common, we show 
$$\langle \varphi_1,\varphi_2\rangle \in \SAT^{{\rm CF}}(\stBool) \times \oSAT^{{\rm CF}} (\stBool)\mbox{ \ \ iff \ \  } 
(\alpha \wedge \varphi_1) \vee \varphi_2 \in \SATPOS \setminus \SAT(\alg{A}).$$

Assume 
$\langle \varphi_1,\varphi_2\rangle \in \SAT^{{\rm CF}}(\stBool) \times \oSAT^{{\rm CF}} (\stBool)$,
i.e., $\varphi_1$ is classically satisfiable, while $\varphi_2$ is not.
Then one can choose an assignment $e_1$ in $\alg{A}$ such that
$e_1(\alpha)>0^\alg{A}$, 
$e_1(\varphi_1)=1^\alg{A}$, and
$e_1(\varphi_2)=0^\alg{A}$. 
Then $e_1((\alpha \wedge \varphi_1) \vee \varphi_2) =e_1(\alpha)$, so $(\alpha \wedge \varphi_1) \vee \varphi_2 \in \SATPOS(\alg{A})$.
On the other hand, by Lemma \ref{lm_SAT_Bool}, $\varphi_2\not\in\SATPOS(\alg{A})$, so $\varphi_2$ is zero under any assignment in $\alg{A}$.
Because $\alpha$ is not fully satisfiable in $\alg{A}$, neither is the term $(\alpha \wedge \varphi_1) \vee \varphi_2 $. 

In the remaining cases, we assume 
$\langle \varphi_1,\varphi_2\rangle \not\in \SAT^{{\rm CF}}(\stBool) \times \oSAT^{{\rm CF}} (\stBool)$
for a pair of CF-terms $\langle \varphi_1, \varphi_2\rangle$. In particular:

Assume $\varphi_2$ is classically satisfiable, no matter what $\varphi_1$ is.
Then one can choose an assignment $e_2$ in $\alg{A}$ such that 
$e_2(\varphi_2) = 1^\alg{A}$. 
Then $e_2((\alpha \wedge \varphi_1) \vee \varphi_2) = 1^\alg{A}$. 

Assume that 
neither $\varphi_1$ nor $\varphi_2$ are classically satisfiable.
This entails that neither term is positively satisfiable in $\alg{A}$, using 
Lemma \ref{lm_SAT_Bool}.
Therefore, for each $e$ in $\alg{A}$, we have that $e(\alpha \wedge \varphi_1)=0^\alg{A}$
and $e(\varphi_2)=0^\alg{A}$. Thus $(\alpha \wedge \varphi_1) \vee \varphi_2$ is
not positively satisfiable in $\alg{A}$.

\end{proof}

\section{General remarks on satisfiability}
\label{section:generalOnSAT}

In this section, we discuss the relation of satisfiability problems
in general $\FLew$-algebras to other notions in logic and applications.

An $\FLew$-algebra $\alg{A}$ and an $\FLew$-term $\varphi$
together define a \emph{function} $\varphi^\alg{A}$.
In investigating satisfiability and positive satisfiability of $\varphi$ in $\alg{A}$,
we are asking two particular questions about the range of $\varphi^\alg{A}$.
In the interpretation given by the two-element Boolean algebra, 
all functions in the algebra are term-definable,
tautologousness and satisfiability of terms are decidable problems, and
the membership of a term in either of the two sets 
is  quite informative about the range of the function defined by the term.
In a general $\FLew$-algebra $\alg{A}$, none of the above is the case:
not all functions on $\alg{A}$ need to be term-definable;
satisfiability and positive satisfiability
capture comparatively less information about the range of the defined function\footnote{We are making 
two comparisons against constants provided by the language; 
other meaningful comparisons are tautologousness and positive tautologousness, defined by
H\'ajek: a term $\varphi$ is a positive tautology in $\alg{A}$ iff $\varphi^\alg{A}$ never assumes the value $0^\alg{A}$.};
if $\alg{A}$ is infinite, there is no obvious algorithm to decide satisfiability or tautologousness in $\alg{A}$. 

\emph{Consistency.} 
A theory $T$ in a logic $\logic{L}$ extending $\FLew$ is a set of $\FLew$-terms that is closed under
deduction over $\logic{L}$. $T$ is consistent iff it does not contain $0$.
We note that $T$ is consistent iff there is a nontrivial $\logic{L}$-algebra $\alg{A}$ and 
an assignment $e_\alg{A}$ such that $e_\alg{A}(\varphi)=1^\alg{A}$ for each $\varphi\in T$ (a nontrivial \emph{model} of $T$).
Indeed, for a consistent $T$, one can get a nontrivial model by considering the Lindenbaum-Tarski algebra of $T$ and
the assignment sending all elements of $T$ to the top element of the algebra; consistency ensures the algebra is nontrivial.
On the other hand, deduction preserves full satisfiability in an algebra:
 if $e(\phi)=1$ and $e(\phi\to\psi)=1$,  then $e(\psi)=1$.
Thus, consistency of $T$ coincides with its satisfiability in a nontrivial $\FLew$-algebra.

\emph{Normal and subnormal functions.} 
In fuzzy mathematics, it is sometimes useful to consider normal sets and normal functions.  
Given a universe $U$, a function $f\colon U \to \alg{A}$ is \emph{normal} iff $f(u)=1^\alg{A}$ for some $u\in U$;
otherwise it is \emph{subnormal}.
Often the function is a membership function of a set, then this terminology also applies to the set itself.  
Clearly, satisfiable terms define normal functions and
 positively satisfiable terms define functions that are not identically zero.

\emph{Solvability of equations.} 
For a given $\FLew$-algebra $\alg{A}$, and two terms $\varphi$ and $\psi$,
one may be interested in the problem whether the equation 
$\varphi \approx \psi$ has a solution in $\alg{A}$.
Such an equation is solvable in $\alg{A}$
iff the equivalence $\varphi\equiv \psi$ is satisfiable in $\alg{A}$, 
and on the other hand, a term $\varphi$ is satisfiable in $\alg{A}$
iff the equation $\varphi\approx 1$ is solvable in $\alg{A}$. 
One can extend this to finite sets of equations by considering their conjunction.
As already remarked, both problems are fragments of the existential theory of $\alg{A}$.

\emph{Definability.}
Let $\alg{A}$ be an $\FLew$-algebra and let $a\in A$.
The value $a$ is implicitly definable by a term $\varphi(x, \bar y)$ in $\alg{A}$
(in the variable $x$) 
iff $\varphi$ is satisfiable in $\alg{A}$ and, for any assignment $e_\alg{A}$ in $\alg{A}$,
if $e_\alg{A}(\varphi)=1^\alg{A}$, then $e_\alg{A}(x)=a$.
For example, it is not difficult to see that all rationals within $[0,1]$ are
implicitly definable in the standard MV-algebra $\standardL$ (cf.~\cite{Hajek:1998}).
Definable values can be used to introduce more general types of satisfiability.
To continue the example of a standard MV-algebra $\standardL$, 
for any term $\phi(\bar x)$, an assignment $e$ in $\standardL$ such that $e(\phi) \ge 1/2$ exists 
iff the term $(y \equiv \neg y) \cdot (y \to \phi(\bar x))$ is fully satisfiable.

\medskip

\section{Conclusion and further research}
\label{section:concluding}

This work investigated satisfiability of $\FLew$-terms in $\FLew$-algebras.
It identified $\WCon$-algebras as the subvariety of $\FLew$-algebras whose nontrivial members 
have classical positive satisfiability;
it characterized classical satisfiability by means of another property, namely, 
the existence of a two-element congruence.
It discussed inclusion order of satisfiability and positive satisfiability problems for $\FLew$-algebras.
It has shown that there are many different satisfiability and positive satisfiability problems in $\FLew$-algebras,
and therefore, most of them are undecidable. 
For any nontrivial $\FLew$-algebra, its satisfiability and positive satisfiability problems
are $\NP$-hard, while the set of positively, but not fully satisfiable terms, if nonempty, is hard for the class $\DP$.

\medskip

We conclude by pointing out some related topics that have not been addressed by this paper, 
and some possibilities of further research into problems investigated here.

One might consider (positive) satisfiability of sets of terms.
Let $T$ be a set of terms. 
If we take `$T$ is (fully) $\alg{A}$-satisfiable' 
to mean `there is an $\alg{A}$-assignment that (fully) satisfies all terms in $T$', 
then for a finite $T=\{\varphi_1, \varphi_2, \dots, \varphi_n\}$, 
we have that $T$ is $\alg{A}$-satisfiable iff so is $\varphi_1 \cdot \varphi_2\cdot \dots \cdot \varphi_n$; 
thus our approach covers finite theories under the given definition of satisfiability of sets of terms. 

On a similar note, one might discuss satisfiability of terms not in a single algebra, but in classes thereof.
This has been initiated in  \cite{Hanikova:Handbook}.

This or other works on satisfiability may prompt a look at 
properties of term-definable functions in $\FLew$-algebras or some related classes. For example, 
one might wonder about term definability for fragments of the algebraic language.
An example of an application of a language fragment is in Section \ref{section:CNFandDNF}.
Logics in fragments of language are often studied in detail (for example, BCK as the
pure implication fragment of $\FLew$) and looking at satisfiability for language
fragments is a natural counterpart.
 
One might propose a different definition of satisfiability: an $\FLew$-term $\varphi$ is satisfiable in an $\FLew$-algebra $\alg{A}$
iff $\sup \{ e_\alg{A}(\varphi) \}=1^\alg{A}$. This definition is inspired by the semantics
of the existential quantifier in first-order extensions of $\FLew$ (cf.~\cite{Rasiowa-Sikorski:MathematicsMetamathematics, Hajek:1998}). 
In general, this would lead to a different set of satisfiable terms in infinite $\FLew$-algebras: 
for example, the term $(x + x) \cdot (\neg x + \neg x)$ is satisfiable in $\standardL$ by a single element, namely, $1/2$;
it is unsatisfiable in a dense subalgebra of $\standardL$ not containing $1/2$, 
however the supremum of values the term takes in such an algebra is $1$.

One might think of suitable generalizations of the issues studied in this paper. 
A possible generalization lies in considering a broader class of algebras. 
Commutativity of multiplication might be dropped, which leads to $\FL{w}$-algebras; 
as these algebras are bounded, both positive and full satisfiability are meaningful notions quite
in the classical vein. The definition of full satisfiability also makes sense for $\FL{i}$-algebras.
In more generality, one might investigate existential theory of (classes of) $\FLew$-algebras or its syntactic fragments.

\bigskip
\noindent{\bf Unsolved problems.}
\begin{itemize}
\item
Assume $R$ is an algorithmically decidable set of primes and consider $\SAT(R^\ast)$
as defined in the proof of Theorem \ref{theorem:MV_SATs}. Is $\SAT(R^\ast)$ decidable?

\item
Is there a nontrivial $\FLew$-algebra $\alg{A}$, such that for every
nontrivial $\FLew$-algebra $\alg{B}$, we have $\SAT(\alg{B}) \subseteq \SAT(\alg{A})$?
\end{itemize}

\bigbreak
{\bf Acknowledgements.} The authors were supported by CE-ITI and GA\v CR
under the grant number GBP202/12/G061 and by RVO:67985807.
The authors are indebted to Rostislav Hor\v c\'{\i}k and an anonymous reviewer for their comments.

\end{document}